\documentclass[10 pt]{article}

\usepackage{amsthm,amsmath,amssymb,epsf, epsfig, color,enumerate,amssymb,hyperref,ytableau}

\usepackage{tikz}

\newcommand{\G}{{\cal G}}
\newcommand{\EX}{{\cal X}}
\newcommand{\J}{{\cal J}}
\newcommand{\N}{{\mathbb{N}}}
\newcommand{\R}{\mathbb{R}}
\newcommand{\Ho}{{\cal H}}
\newcommand{\M}{{\cal M}}
\newcommand{\Lc}{{\cal L}_3}

\newtheorem{defin}{Definition}[section]
\newtheorem{thm}{Theorem}[section]
\newtheorem{prop}[thm]{Proposition}
\newtheorem{claim}[thm]{Claim}
\newtheorem{remark}[thm]{Remark}
\newtheorem{lemma}[thm]{Lemma}
\newtheorem{con}[thm]{Corollary}

\newcommand{\arx}[1]{\href{http://arxiv.org/abs/#1}{\texttt{arXiv:#1}}}

\date{}

\title{Mixing Time for Some Adjacent Transposition Markov Chains}

\author{Shahrzad Haddadan, Peter Winkler}

\begin{document}

\maketitle

\begin{abstract}

We prove rapid mixing for certain Markov chains on the set $S_n$ of permutations on $1,2,\dots,n$ in which adjacent transpositions
are made with probabilities that depend on the items being transposed.  Typically, when in state $\sigma$, a position $i<n$ is chosen
uniformly at random, and $\sigma(i)$ and $\sigma(i{+}1)$ are swapped with probability depending on $\sigma(i)$ and $\sigma(i{+}1)$.
The stationary distributions of such chains appear in various fields of theoretical computer science \cite{Wilson, Self1, Mallow},
and rapid mixing established in the uniform case \cite{Wilson}.

Recently, there has been progress in cases with biased stationary distributions \cite{Benjamini, Dana}, but there are wide classes of
such chains whose mixing time is unknown.  One case of particular interest is what we call the ``gladiator chain,'' in which each number $g$
is assigned a ``strength'' $s_g$ and when $g$ and $g'$ are adjacent and chosen for possible swapping, $g$ comes out on top with probability
$s_g/(s_g + s_{g'})$.   We obtain a polynomial-time upper bound on mixing time when the gladiators fall into only three strength classes.

A preliminary version of this paper appeared as ``Mixing of Permutations by Biased Transposition'' in STACS 2017 \cite{Stacs}.
 \end{abstract}

\section{Introduction}\label{intro}
For $n\in \N$, let $S_n$ be the set of all permutations of the numbers $1,2,\dots ,n$.  One can think of a permutation as the order
in which a search engine arranges its results \cite{Mallow}, the order in which a self organizing list arranges its items \cite{Self1,Self2},
or the order the playing cards appear after shuffling \cite{Wilson, Diaconis}; each of these suggests different probability distributions on $S_n$.
Taking samples from such distributions is a useful task which can be tackled using a Markov chain, in particular when dynamic programming approaches
fail to have a polynomial runtime.\footnote{We note here that for the particular case of our study, i.e., gladiators with constant number of
strengths, the dynamic programming approach is efficient. However, the mixing problem is still interesting for at least two reasons:
(1) As discussed in the introduction, a self-organizing list is basically a Markov chain with high mixing time. Thus, analyzing the gladiator
chain is closely related to studying this data structure's performance.  (2) Dynamic programming algorithms would require exponential
time when we have polynomial number of teams. Thus, employing Markov chains could provide an efficient sampling tool in such cases.} 

\smallskip
A natural Markov chain on $S_n$ picks a number $1\leq i\leq n{-}1$ uniformly at random and from state $\sigma$, puts $\sigma(i{+}1)$
ahead of $\sigma(i)$ with probability $p_{\sigma(i),\sigma(i+1)}$. We call such chains \emph{adjacent transposition} Markov chains. 

\medskip 

In this paper, we consider the total variation mixing time, which is defined as the number of steps required before the total variation distance
between the distribution of the current state and stationarity is less than $\epsilon$ (where $\epsilon$ is some fixed convergence factor).
For Markov chain $\cal M$ we denote this time by $t_{\epsilon}({\cal M})$, or if $\epsilon=1/4$, simply by $t({\cal M})$.
 
\smallskip

Jim Fill \cite{FillConj} conjectured that if an adjacent transposition Markov chain is monotone, then it is rapidly mixing.
Monotonicity in this context means that for all $i,j$ satisfying $1\leq i <j  \leq n$, $p_{i,j}\geq 1/2$, $ p_{i,j-1}\leq p_{i,j}$, and 
$p_{i+1,j}\leq p_{i,j}$  \cite{FillConj}. Furthermore, his conjecture asserts ``the simple chain'' whose stationary distribution is uniform
has the highest spectral gap among all monotone adjacent transposition chains.\footnote{ The spectral gap is another measure of mixing.
Here, we are interested in total variation mixing time which, in this case, is within a polynomial factor of the spectral gap.}

\smallskip

Here we provide a brief history of the results  on the adjacent transposition Markov chains.   All of these chains are monotone and rapidly mixing.
Wilson and Benjamini's papers \cite{Wilson, Benjamini} led to Fill's conjecture \cite{FillConj}; Bhakta et al.\ \cite{Dana} verified the conjecture
in two cases.  The current paper, as well as a recent result by Miracle et al. \cite{SaraAmanda}, study the so-called ``gladiator chain'' under
certain conditions, and verify Fill's conjecture in limited cases. We will define the gladiator chain and present a few of its applications
later in this introduction. 

\medskip

 \textbf{1. The simple chain. } In the case where $p_{i,j}=1/2$ for all $i$ and $j$, the chain will have a simple description:
Given a permutation $\sigma$, pick two adjacent elements uniformly at random, and flip a fair coin to decide whether to swap them.
We call this chain, whose stationary distribution is uniform, the \emph{simple} chain. Getting precise mixing results for this chain turned
out not to be simple; many papers targeted this problem \cite{Simple1,ComparisonMethodDiaconis}, and finally Wilson \cite{Wilson} showed
the mixing time for this chain is $\Theta(n^3 \log n)$ (that is, he obtained lower and upper bounds within a constant factor).
 
 \smallskip
\textbf{2. The constant-bias chain.} After Wilson's paper, Benjamini et al. \cite{Benjamini} studied the case where $p_{i,j}=p>1/2$ for all $i<j$,
and $p_{j,i}=1{-}p$.  The stationary distribution of this chain is the one assigning a probability  proportional to $p^{inv(\sigma)}$,
to each $\sigma\in S_n$ where $inv(\sigma)$ is the number of inversions in $\sigma$. This distribution appears in statistics and machine
learning since it is the distribution generated by the ``Mallows model'' \cite{Mallow, Mallow2}.

Benjamini et al. \cite{Benjamini}, showed that the constant biased Markov chain is closely related to another Markov chain known as
the \emph{asymmetric simple exclusion process}, and both chains mix in $\Theta(n^2)$ steps. We will talk more about exclusion processes
later on in this introduction.

 \smallskip

\textbf{3. ``Choose your weapon" and ``league hierarchy" chains.} The following two special cases were studied by Bhakta et al. \cite{Dana}.
In the \emph {choose your weapon chain} $p_{i,j}$ is only dependent on $i$, and the \emph{league hierarchy chain} is given by a binary tree
$T$ with $n$ leaves. Each interior node $v$ of $T$ is labeled with some probability $1/2 \leq q_v\leq 1$, and the leaves are labeled by numbers
$1\dots n$. The probability of putting $j$ ahead of $i$ for $j>i$ is equal to $p_{i,j}=q_{ j \wedge i}$ where $j\wedge i$ is the node that
is the lowest common ancestor of $i$ and $j$ in $T$. Bhakta et al.\ showed that the choose your weapon chain mixes in ${\cal O}(n^8\log n)$
steps and the league hierarchy chain in ${\cal O}(n^4 \log n)$ steps. 
\medskip

Here we are interested in \emph{gladiator} chains, which constitute a subclass of the monotone adjacent transposition chains. Gladiator chains
have connections to self organizing lists, and were introduced by Jim Fill.
\medskip

Fill was interested in probabilistic analysis of algorithms for \emph{self-organizing lists} (SOLs). Self-organizing lists are data structures 
that facilitate linear searching in a list of records; the objective of a self-organizing list is to sort the records in non-decreasing order
of their access frequencies \cite{Self1}. Since these frequencies are not known in advance, an SOL algorithm aims to move a particular
record ahead in the list when access on that record is requested. There are two widely used SOL algorithms: the \emph{move ahead one}
algorithm (MA1) and the \emph{move to  front }algorithm (MTF). In MA1, if the current state of the list is
$(x_1,x_2,\dots,x_{i-1}, x_i,x_{i+1},\dots , x_n)$ and the $i$th record is requested for access,  it will go ahead in the list only one
position and the list will be modified to $(x_1,x_2,\dots,x_i,x_{i-1},x_{i+1},\dots, x_n)$.  In MTF it will go to the front and the
list will be modified to $(x_{i},x_1,x_2,\dots,x_{i-1},x_{i+1},\dots, x_n)$. It appears that  MA1 should perform better than MTF when the
list is almost sorted and worse when the low frequency records are standing in front; although this has been confirmed by simulations,
it has not been analytically confirmed \cite{Self2}.  Considering the adjacent transposition Markov chain corresponding to MA1, Fill shows
\cite{FillConj} that there are cases in which the chain is not rapidly mixing. Hence, he poses the question of sampling from the stationary
distribution of MA1, and he introduces the gladiator chain which has the same stationary distribution as MA1 and seems to be rapidly mixing for arbitrary
choice of parameters.

In the gladiator chain, each element $i$ can be thought of as a gladiator with strength $s(i)$. Every permutation of numbers $1,2,\dots n$
can be thought of as a ranking of gladiators. In each step of Markov chains we choose $1\leq k < n$ uniformly at random, i.e., we choose
adjacent gladiators $\sigma(k)=i$ and $\sigma(k+1)=j$. These gladiators will fight over their position in ranking. With probability
$p_{j,i}=s(i)/(s(i)+s(j))$, gladiator $i$ will be the winner of the game and will be placed ahead of $j$ in $\sigma$ if he isn't already.
With probability $1{-}p$, $j$ is put ahead of $i$.  If Fill's conjecture holds, gladiator chains must mix rapidly. 
 \medskip

\textbf{Exclusion processes.} A related Markov chain which has received a lot of attention is the \emph{exclusion process} (\cite{exc1, exc2}).
In this chain we have a graph  $G= \langle V, E\rangle$  and  $m< \vert V\vert$ particles  on the vertices of $G$. The sample space is the
set containing all the different placements of the $m$ particles on vertices of $G$. At each step of the Markov chain we pick a vertex $v$
uniformly at random with probability  $1/\vert V \vert$ and  one of its adjacent vertices $w$ with probability  $1/d(v)$. If there is a
particle in one of the vertices and not the other one, we swap the position of the particle with a constant probability $p$. We are interested
in the \emph{linear} exclusion process when the graph is a finite path with $n$ vertices.  As mentioned before, the linear exclusion process
was studied by Benjamini et al.\ \cite{Benjamini} and is known to be mixing in time $\Theta(n^2)$.\footnote{Benjamini et al.\ use this
result to prove that the constant biased adjacent transposition chain is rapidly mixing.} Later, Greenberg et al.\ \cite{Greenberg}
presented a simpler proof. 

 \medskip
\textbf{Our Contribution.}
We study the gladiator chain when the gladiators fall into a constant number of teams, gladiators in each team having the same strength
(Definition \ref{defteams}).  We then extend the definition of linear exclusion process (studied by Benjamini et al.) by allowing particles
of different types to swap their positions on a line. We call this new chain a \emph{linear particle system} (Definition \ref{defparts}).
We will show that mixing results for linear particle systems can produce mixing results for gladiator chains (Theorem \ref{reduction}). 

In particular, we study the linear particle system in which there are three particle types, and in Theorem \ref{mainthm} we extend
Benjamini et al.'s result by  showing the three particle system mixes rapidly. Having Theorem \ref{mainthm} we conclude that the
following adjacent transposition chains mix rapidly, and hence confirming Fill's conjecture in these cases: The gladiator chain
when gladiators fall into three teams of same-strength gladiators; and the league hierarchy chain for ternary trees
(extending Bhakta et al.'s work \cite{Dana}). 
\medskip

\textbf{Remark.} We believe linear particle systems, like exclusion processes, are interesting Markov chains that may appear as components
of other Markov chains, and thus would facilitate studying mixing times of other chains. For instance, in Section \ref{trees} of this paper,
by using Theorem \ref{mainthm} we extend a result about binary trees to ternary trees.  As another example, we remind the reader of the
correspondence between the exclusion process and the Markov chains on the lattice paths in an $n\times m$ rectangular lattice
(Figure \ref{figlatticepath}). Similarly, there is a correspondence between the linear particle systems having $k$ particles and the lattice
paths in $k-$dimensional lattices (Figure \ref{figlatticepath}). Some Markov chains defined on lattice paths in a $k-$dimensional rectangle
have already been studied by Greenberg et al.\ \cite{Greenberg}.

\smallskip

We remark here that following our result in STACS 2017 \cite{Stacs}, Miracle et al.\ \cite{SaraAmanda} studied the mixing time of
linear particle system when the number of particles is a constant $k$, and showed the mixing time is upper bounded by $n^{2k+4}$.
With different techniques from ours, they  prove the mixing time of gladiator chains with a constant number of teams is upper bounded
by $n^{2k+6}\log k$. The mixing time for linear particle systems and gladiator chains with teams, remains an open problem in the cases
in which the number of particle types or teams is more than a constant. 

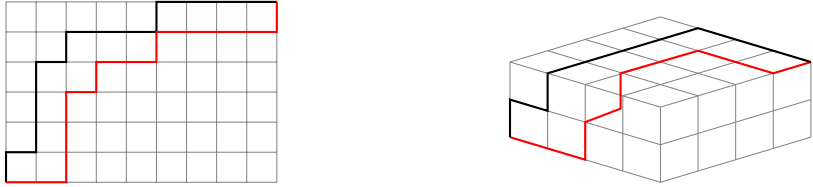
\begin{figure}[!h]\label{figlatticepath}
\smallskip
\begin{tikzpicture}
\draw[step=0.4cm,gray,very thin] (-1.2,-1.2) grid (2.4,1.2);
\draw[thick](-1.2,-1.2) --(-1.2,-0.8)--(-0.8,-0.8)--(-0.8,-0.4)--(-0.8,0.0)--(-0.8,0.4)--(-0.4,0.4)
--(-0.4,0.8) --(0.0,0.8)--(0.4,0.8)--(0.8,0.8)--(0.8,1.2)--(1.2,1.2)--(1.6,1.2)--(1.8,1.2)--(2.4,1.2) ;
\draw[red,thick](-1.2,-1.2) --(-0.8,-1.2)--(-0.4,-1.2)--(-0.4,-0.8)--(-0.4,-0.4)--(-0.4,0.0)--(0.0,0.0)
--(-0.0,0.4) --(0.4,0.4)--(0.8,0.4)--(0.8,0.8)--(1.2,0.8)--(1.6,0.8)--(2.0,0.8)--(2.4,0.8)--(2.4,1.2);
\draw[step=0.4cm, gray, very thin] (5.5,-0.6) -- (7.5,-1.2)--(9.5,-0.6);
\draw[step=0.4cm, gray, very thin] (5.5,-0.1) -- (7.5,-0.7)--(9.5,-0.1);
\draw[step=0.4cm, gray, very thin] (5.5,0.4) -- (7.5,-0.2)--(9.5,0.4);
\draw[step=0.4cm, gray, very thin] (5.5,0.4) -- (7.5,1.0)--(9.5,0.4);
\draw[step=0.4cm, gray, very thin]  (6.0,0.55)-- (8.0,-0.05);
\draw[step=0.4cm, gray, very thin]  (6.5,0.70)-- (8.5,0.10);
\draw[step=0.4cm, gray, very thin]  (7.0,0.85)-- (9.0,0.25);
\draw[step=0.4cm, gray, very thin]  (6.0,0.25)--(8.0,0.85);
\draw[step=0.4cm, gray, very thin]  (6.5,0.1)--(8.5,0.7);
\draw[step=0.4cm, gray, very thin]  (7.0,-0.05)--(9,0.55);

\draw[step=0.4cm, gray, very thin]  (7.5,-0.2)--(7.5,-1.2);
\draw[step=0.4cm, gray, very thin]  (7.0,-0.05)--(7.0,-1.05);
\draw[step=0.4cm, gray, very thin]  (8.0,-0.05)--(8.0,-1.05);
\draw[step=0.4cm, gray, very thin]  (6.5,0.10)--(6.5,-0.90);
\draw[step=0.4cm, gray, very thin]  (8.5,0.10)--(8.5,-0.90);
\draw[step=0.4cm, gray, very thin]  (6.0,0.25)--(6.0,-0.75);
\draw[step=0.4cm, gray, very thin]  (9.0,0.25)--(9.0,-0.75);
\draw[step=0.4cm, gray, very thin]  (5.5,0.40)--(5.5,-0.60);
\draw[step=0.4cm, gray, very thin]  (9.5,0.40)--(9.5,-0.60);
\draw[red,thick](5.5,-0.6)--(6.0,-0.75)--(6.5,-0.90)--(6.5,-0.4)--(6.97,-0.22)--(6.97,0.25)--
(8,0.55)--(9.0,0.25)--(9.5,0.40);
\draw[black,thick](5.5,-0.6)--(5.5,-0.1)--(6.0,-0.25)--(6.0,0.25)--(8.0,0.85)--(9.5,0.40);
\end{tikzpicture}
\caption{The correspondence between lattice paths and linear particle systems:
The  picture on the left  illustrates two paths in a two dimensional lattice; the red one corresponds to
001110100100001 and the black one corresponds to 10110100010000.
The picture on the right illustrates two paths in a three dimensional lattice; the red one corresponds to
0012122002 and the black one corresponds to 1012222000.
}
\end{figure}

\medskip
Definitions and results are presented in Section \ref{Section1}, along with the correspondence between the gladiator chains and the
linear particle systems. Section \ref{proof} contains the proof that the linear three-type system mixes rapidly under certain conditions.
In Section \ref{trees}, we discuss the league hierarchy chain and our result for ternary trees.  

\section{Definitions and Results}\label{Section1}

\begin{defin}\textbf{Gladiator chain} (Playing in teams).\label{defteams}
Consider the Markov chain  on state space $S_n$ that has the following properties: The set  $[n]$ (i.e. gladiators)  can be  partitioned
into subsets: $T_1 , T_2, \dots , T_k $ ($k$ teams). We have the following strength function: $s:[n]\rightarrow \R$, $s(g)=s_j$ iff $g\in T_j$.
At each step of Markov chain,  we choose $i\in [n{-}1]$ uniformly at random. Given that we are at state $\sigma$, and $\sigma(i)=g,
\sigma(i{+}1)=g'$, we put $g$ ahead of $g'$ with probability $\frac{s(g)}{s(g)+s(g')}$.  We denote a gladiator chain having $n$ gladiators
playing in $k$ teams by ${\G}_k(n)$.\footnote{Although the notation $\G_k(n_1,n_2,\dots,n_k)$ would be more precise ($n_i$ being cardinality
of $T_i$), we avoid using it for simplicity and also because our analysis is not dependent on $n_1, n_2, \dots, n_k$.}

\end{defin}

This is a reversible Markov chain and the stationary distribution $\pi$ is
\begin{equation}
\hspace{2 cm}\pi(\sigma)=\prod_{\substack{i=1}}^{n} s(i)^{\sigma^{-1}(i)}/Z.
\quad \quad(Z \text{ is a normalizing factor.})
\end{equation}\label{stationary}

Note that by writing $\sigma(i)=g$ we mean gladiator $g$ is located at position $i$ in $\sigma$. By writing 
$\sigma^{-1}(g)$ we are referring to the position of gladiator $g$ in the permutation $\sigma$. We use this notation throughout the text
and for permutations presenting both gladiators and particles.

\begin{defin}\textbf{Linear particle systems.}\label{defparts}
Assume we have $k$ types of particles and of each type $i$, we have $n_i$ indistinguishable copies. Let $n=\sum_{i=1}^{k}n_i$. 
Let $\Omega_{n_1,n_1,\dots n_k}$ be the state space containing all the different linear arrangements of these $n$ particles. 
If the current state of the Markov chain is $\sigma$, choose $i\in [1,n-1]$ uniformly at random.  Let $\sigma(i)$ be of type $t$
and $\sigma(i+1)$ be of type $t'$. If $t=t'$ do nothing. Otherwise, put $\sigma(i)$ ahead of $\sigma(i+1)$ w.p.\ $p_{t,t'}$ and
put $\sigma(i+1)$ ahead of $\sigma(i)$ w.p.\ $1-p_{t,t'}$.  We denote the linear particle system having $n$ particles of
$k$ different types by ${\EX}_k(n)$.
\end{defin}
This chain  is also a reversible Markov chain. In the special case where $p_{t,t'}=\frac{s(t)}{s(t)+s(t')}$  the stationary distribution $\pi$ is
\begin{equation}
\hspace{2 cm}\pi(\sigma)=\prod_{\substack{i=1}}^{n} s(i)^{\sigma^{-1}(i)}/Z'.
\quad \quad(Z' \text{ is a normalizing factor.})
\end{equation}\label{eq2}

\begin{prop}
By regarding gladiators of equal strength as indistinguishable particles, we associate to any gladiator system a linear particle system.
\end{prop}

Note that the state space of the gladiator system has cardinality $n!$ for $n$ different gladiators but the linear particle system has only
$n!/(n_1!n_2!\dots n_k!)$ states, since particles of the same type are indistinguishable. Thus, $Z'\ll Z$. The following theorem,
whose proof will be presented later, shows the connection between the mixing times of the two chains.
\begin{thm}\label{reduction}
Let $t({\EX_k})$ and $t({\G}_k)$ be respectively the mixing times for a linear particle system and its corresponding gladiator chain.
Then, $t(\G_k) \leq \mathcal{O} (n^8)~ t(\EX_k)$.
\end{thm}

Our main result, which extends the results of Benjamini et al.\ \cite{Benjamini} on exclusion processes, is the following:

\begin{thm}\label{mainthm}
Let $\EX_{3}(n)$ be a linear particle system of Definition \ref{defparts}, having particles of type A, B and C.
Assume that we have strength functions assigned to each particle type, namely $s_A< s_B<s_C$, and thus swapping probabilities
$p_{B,A}=s_A/(s_A+s_B)$, $p_{B,C}=s_C/(s_C+s_B)$ and $p_{A,C}=s_C/(s_A+s_C)$.
If $s_A/s_B,s_B/s_C<1/2$, then the mixing time of $\EX_{3}(n)$ satisfies $t({\EX}_{3}(n))\leq {\cal O}(n^{10}).$

\end{thm}

\begin{remark}
The condition $s_A/s_B,s_B/s_C\leq 1/2$ comes from the following simple bound on $q$-binomials that we later prove in Lemma~\ref{qlemma}: 
If  $q<1/2$ then, 
${{m}\choose{r} }_q< 2^r < (\frac{1}{q})^r.$
Better bounds on $q$-binomials would allow the result to be improved.
\end{remark}

We will prove Theorem \ref{mainthm} in Section \ref{proof}.
Having Theorem \ref{mainthm}, we deduce the following case of Fill's conjecture: 
 
 \begin{con}
The mixing time of ${\G}_3(n)$ satisfies $t(\G_{3}(n))\leq {\cal O}(n^{18})$, provided $s_{A}/s_{B}<1/2$ and
$s_{B}/s_{C}<1/2$, where $C$ is the strongest playing team among the three, and  the gladiators in team $B$ are stronger than the gladiators in team $A$.
\end{con} 
 
\begin{proof}
From Theorems \ref{mainthm} and \ref{reduction}. 
\end{proof}
 
We present the following corollary of Theorem \ref{mainthm} here and discuss it in full detail later in Section \ref{trees}.
 
\begin{con}(League hierarchies for ternary trees)\label{LeagueHi}
Let $T$ be a ternary tree  with $n$ leaves. The children of each interior node $v$ are labeled with labels $A(v)$, $B(v)$, and $C(v)$,
and each internal node has three strength values $s_{A(v)}$, $s_{B(v)}$, and $s_{C(v)}$. The leaves are labeled by numbers $1,2,\dots ,n$.
The probability of putting $j$ ahead of $i$ for $j>i$ is equal to $p_{i,j}=s_{X(v)}/(s_{X(v)}+s_{Y(v)})$ where $v$ is the node
that is the lowest common ancestor of $i$ and $j$ in $T$, and $X(v)$ is the child of $v$ which is an ancestor of $j$,
and $Y(v)$ is the child of $v$ which is an ancestor of $i$. If for each $v\in T$, $s_{A(v)}$, $s_{B(v)}$, and $s_{C(v)}$
satisfy the conditions in Theorem \ref{mainthm}, then the mixing time of the league hierarchy chain is bounded by $n^{14}\log n$.
\end{con}

 \medskip
  
We finish this section by proving Theorem \ref{reduction}.
 
\subsection{Gladiators and Particles (Proof of Theorem \ref{reduction})}

Consider the gladiator chain ${\G_{k}(n)}$ for arbitrary $n$ being the number of gladiators  and $k$ the number of teams.
Assume that we have ${n}_i$ gladiators on team $i$; hence, $\sum_{i=1}^k{n}_i=n$.  At each step of the chain, one of two things is happening: 
\begin{enumerate}
\item Whisking: gladiators of the same team are fighting. 
\item Sifting:  gladiators of different teams are fighting.
\end{enumerate}
If we were restricted to  whisking steps the chain would be equivalent to a product of several simple chains analyzed by Wilson \cite{Wilson}.
If we were restricted to sifting steps the chain would be the linear particle system chain introduced in Definition \ref{defparts}. 
In order to study the mixing time of the gladiator chain we  analyze sifting and whisking steps separately, and then we employ the
following decomposition theorem: 

\begin{thm}\label{decom}  \textbf{Decomposition Theorem \cite{Decomposition}.}
Let $\cal M $ be a Markov chain on state space $\Omega$ partitioned into $\Omega_{1},\Omega_{2},\dots,\Omega_k$.
For each $i$, let ${\cal M}_i$ be the restriction of $\cal M$ to $\Omega_i$ that rejects  moves  going outside of $\Omega$.
Let $\pi_{i} (A)= \pi(A\cap \Omega_i)/\pi(\Omega_i)$ for $A\subseteq \Omega_i$. We define the Markov chain $\bar{\cal M}$ on
state space $\{1,\dots k\}$ as follows: $Pr_{\bar{\cal M}}(i,j)= \sum_{x\in \Omega_{i},y\in \Omega_j} \pi_i(x)Pr_{\cal M}(x,y)$,
where $Pr_{\cal M}$ and $Pr_{\bar{\cal M}}$ are transition probabilities of $\cal M$ and $\bar{\cal M}$ respectively. 
Then $$t({\cal M})\leq 2 t({\bar{\cal M}}) \max_i\{t({\cal M}_i)\}.$$

\end{thm}
 To apply the decomposition theorem, we partition $S_n$ to $S_{\sigma_1,\sigma_2, \dots, \sigma_k} $ for all choices of $ \sigma_1\in S_{n_1},
\sigma_2\in S_{n_2},\dots , \sigma_k\in S_{n_k}$, each $S_{\sigma_1,\sigma_2, \dots, \sigma_k}$ being the set of all permutations in
$S_n$ in which all the gladiators corresponding to particle $i$ preserve the ordering associated to them by $\sigma_i$.
The restriction of $\G_k(n)$ to $S_{\sigma_1,\sigma_2, \dots, \sigma_k}$ is equivalent to $\EX_{k}(n)$.
We define $\bar{\G}$ to be the Markov chain on $\prod_{i=1}^k S_{n_i}$ with the following transition probabilities: 

$$Pr_{\bar{\G}}(S_{\sigma_1,\sigma_2, \dots,\sigma_i,\dots, \sigma_k},S_{\sigma_1,\sigma_2, \dots,\sigma'_i,\dots, \sigma_k})={ \displaystyle \sum_{\atop{\substack{
x\in S_{\sigma_1,\sigma_2, \dots,\sigma_i,\dots, \sigma_k},\\y\in S_{\sigma_1,\sigma_2, \dots,\sigma'_i,\dots, \sigma_k}}}
} \frac{\pi(x)Pr_{\G}(x,y)}{{\pi(S_{\sigma_1,\sigma_2, \dots,\sigma_i,\dots, \sigma_k})}}},$$
where  $\sigma_i$ and $\sigma_i'$ are only different in swapping  $j$ and $j+1$st elements and
$Pr_{\G}(x,y)=1/2(n-1)$ iff $j$ and $j{+}1$st copies of particle $i$ are adjacent in $x$ and swapped in $y$. 
Moreover, we observe that:
 $$\frac{1}{{\pi(S_{\sigma_1,\sigma_2, \dots,\sigma_i,\dots, \sigma_k})}}{ \displaystyle \sum_{\atop{\substack{
x\in S_{\sigma_1,\sigma_2, \dots,\sigma_i,\dots, \sigma_k},\\y\in S_{\sigma_1,\sigma_2, \dots,\sigma'_i,\dots, \sigma_k}}}
} \pi(x)} \geq1/(n-1).$$
We can verify the above equation by the following reasoning: consider an arbitrary permutation
$z\in S_{\sigma_1,\sigma_2, \dots,\sigma_i,\dots, \sigma_k}$ in which $j$th and $j{+}1$st copies of particle $i$ are not adjacent.
We can map $z$ to two other permutations $z_1$ and $z_2$ where in $z_1$ we take the the $j$th copy of particle $i$ down to make it
adjacent to the $j{+}1$st copy, and in  $z_2$ we take the the $j{+}1$st copy of particle $i$ up to make it adjacent to the $j$th copy.
We will have $\pi(z)/\pi(z_1)=\pi(z_2)/\pi{(z)}$, and hence one of $\pi(z_1)$ or $\pi(z_2)$ will be larger than $\pi(z)$.
This mapping is in worst case $n{-}1$ to $1$, hence the above equation holds.  

Having the above observations, we realize $\bar{\G}$ is the product of $k$ adjacent transposition Markov chains, and in each of these
Markov chains we swap two adjacent elements with probability at least $1/2(n-1)^2$. Let these chains be $\bar{\G}_1,\bar{\G}_2,\dots, \bar{\G}_k$.
By comparing the conductance (for more information about conductance, see \cite{MCBook}) of this chain to the simple chain analyzed
by Wilson \cite{Wilson}, for each $i$ we will have $t(\bar{\G}_i)\leq {n}_i^8$. 
We use the following Theorem of \cite{Dana}:

\begin{thm}\label{product}
If  $\bar{\G}$ is a product of $k$ independent Markov chains $\{\bar{\G}_i\}^{k}_{i=1}$
and it updates each  ${\bar{\G}_i}$ with probability $p_i$, then
$$t_{\epsilon}(\bar{\G})\leq \displaystyle \max_{i=1,\dots, n} \frac{2}{p_i}~ t_{\frac{\epsilon}{2k}}(\bar{\G}_i).$$
\end{thm}

Plugging in $p_i=n_i/n$, we have $t(\bar{\G})\leq \max( 2n/n_i) n_i^8\leq 2 n^8$. Summing up and employing the Decomposition Theorem,

$$t(\G_{k}(n))\leq 4 n^8 t(\EX_{k}(n)).$$

\section{ Three-Particle Systems (Proof of the Main Theorem)}\label{proof}

In this section we prove Theorem \ref{mainthm} which states that $t(\EX_3(n))\leq {\cal O}(n^{10}) $ if 
$s_A/s_B, s_B/s_C\leq 1/2$.

Assume that we have $a$ copies of particle $A$, $b$ copies of particle $B$, and $c$ copies of particle $C$.
We denote the set containing all the different arrangements of these particles by $\Omega_{a,b,c}$.
We introduce another Markov chain $\EX_t(n)$ on the same sample space $\Omega_{a,b,c}$.
Using the comparison method (see \cite{compare}) we will show that the mixing times of $\EX_3(n)$ and $\EX_t(n)$ are related.

Then we will use the path congestion technique to show $\EX_t(n)$ mixes in polynomial time, and hence we deduce Theorem \ref{mainthm}.

$$\text{mixing time of } \EX_3(n) \xleftarrow[\text{ technique}]{\text{Comparison}}\text{mixing time of } \EX_t(n) $$
\textbf{Notation.} We denote the substring $\sigma(i)\sigma(i+1)\dots\sigma(j)$ by $\sigma[i,j]$, and by $B^t$ we refer
to a string which is $t$ copies of particle $B$. 

\begin{defin}
Let  ${\EX}_t (n)$ be a Markov chain on state space $\Omega_{a,b,c}$ and $n=a+b+c$. If the current state is $\sigma$,
we choose natural numbers $1\leq i<j\leq n$ uniformly at random and swap them following these rules (Figure \ref{figJH}):
\end{defin}

\begin{enumerate}
\item If $\sigma(i)=A$  and in $\sigma(j)=C$ or vice versa and $\sigma[i{+}1,j{-}1]=B^{j-i-1}$.  Then,
put $\sigma(i)$ and $\sigma(j)$  in increasing order of their strength w.p.\ $(s_C/s_A)^{(j{-}i)}/(1+(s_{C}/s_A)^{(j{-}i)})$.
With probability $1/(1+(s_{C}/s_A)^{(j{-}i)})$, put them in decreasing order. We call this move a \emph{Jump} and we denote
it by $\J_{i}^{j}(A,C)$ if $\sigma(i)=A$ and $\sigma(j)=C$; and $\J_{i}^{j}(C,A)$ for vice versa. 

\item If $\sigma[i,j{-}1]=B^{j-i}$ and  $\sigma(j)=A $   or  if $\sigma[i+1,j]=B^{j-i}$ and  $\sigma(i)=A $.
Then, put $\sigma(i)$ and $\sigma(j)$ in increasing order of their strength w.p. $(s_{B}/s_A)^{j{-}i}/(1+(s_{B}/s_A)^{j{-}i})$.
With probability $1/(1+(s_{B}/s_A)^{j{-}i})$, put them in decreasing order. We call this move a \emph{Hop},
and we denote it by ${\Ho}_{i}^j(A,B)$  if $\sigma(i)=A$  and $\sigma(j)=B$; and ${\Ho}_{i}^{j}(B,A)$ for vice versa.
Similar rules and notation apply when swapping $B$ and $C$. 
\item Else, do nothing.
\end{enumerate}

\begin{figure}[!h]

\centerline{\resizebox{0.7 \linewidth}{!}{\includegraphics{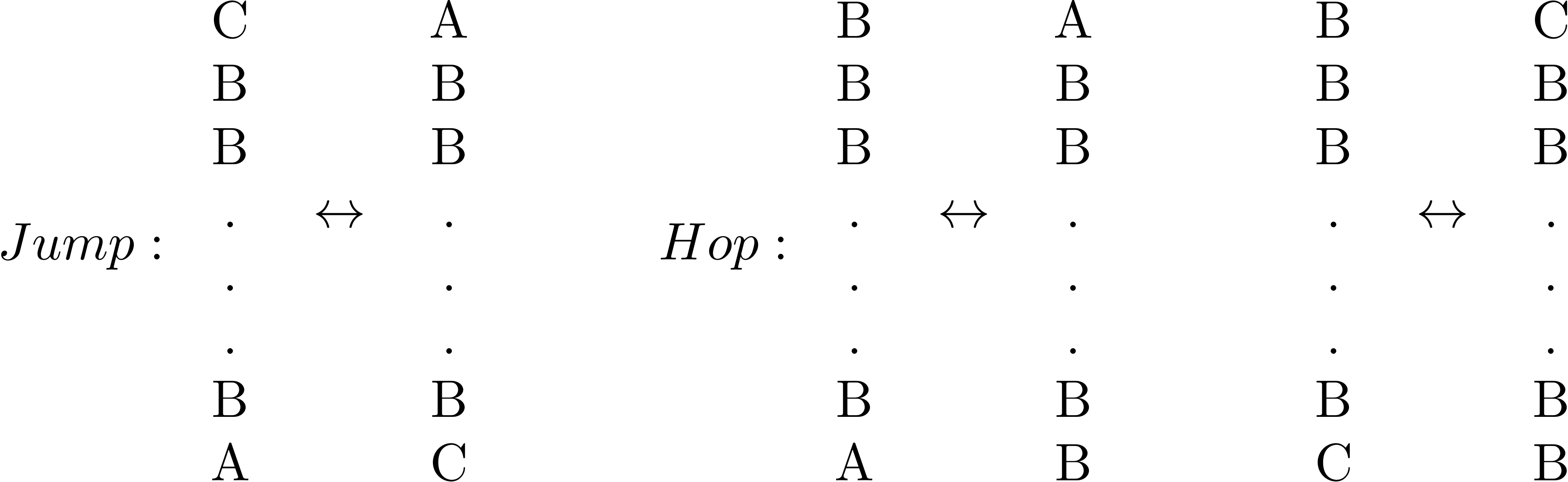}}}
\caption{Jumps and Hops are the transitions  in the Markov chain $\EX_t$.}\label{figJH}
\end{figure}


It can be easily checked that $\EX_t$ is reversible and its stationary distribution is the $\pi$ in Equation \ref{eq2}.

\begin{lemma}\label{lem1}
$t(\EX_3(n))\leq 2n^4 ~ t(\EX_t(n)).$
\end{lemma}

\begin{proof}

We use the comparison technique\footnote{
The comparison method was introduced by Diaconis and Saloff-Coste \cite{ComparisonMethodDiaconis}; Randall and Tetali extended it and employed
it for analysis of Glauber dynamics \cite{compare}.  For more information about this method we encourage the reader to refer to \cite{MCBook}.
}
 in the proof of Lemma \ref{lem1} (see \cite{ComparisonMethodDiaconis, compare}). To any edge $(\sigma,\tau)$ in $\EX_t$, we assign a path
from $\sigma$ to $\tau$ in $\EX_3$.  Let $e_{i}(p,p')$ be a move in $\EX_3$ which swaps particles $p$ and $p'$ located at positions
$i$ and $i+1$ in an arrangement. To  $e= (\sigma,\tau)$ making $\J_{i}^j(A,C)$ in $\EX_{t}$, we correspond the following path in $\EX_3$:
$e_i(A,B) , e_{i+1}(A,B),\dots e_{j-2}(A,B), e_{j-1}(A,C), e_{j-2}(B,C),\dots e_{i}(B,C)  $. We denote this path by $\gamma_{\sigma\tau}$,
and the set contaning all such paths by  $\Gamma_{\J}$. 
Similarly, to  $e= (\sigma,\tau)$ making $\Ho_{i}^j(A,B)$ in $\EX_{t}$, we correspond the following path in $\EX_3$:
$e_i(A,B) , e_{i+1}(A,B),\dots e_{j-2}(A,B), e_{j-1}(A,B) $. We denote this path by $\gamma_{\sigma,\tau}$, and the set contaning all such
paths by $\Gamma_{\Ho}$. Let $\Gamma=\{ \gamma_{\sigma,\tau}\}_{\sigma,\tau\in \Omega_{a,b,c}}=\Gamma_{\J}\cup \Gamma_{\Ho}$.

We now bound the congestion placed by $\Gamma$ on edges of $\EX_3$.
Consider an arbitrary $e=(\alpha,\beta)$  making swap $e_i(A,B)$ and assume $\alpha[i-t-1,i+d+1]=pB^tAB^{d}p'$ where $p$ and
$p'$ are particles different from $B$. For any $\sigma$ and $\tau$ in $\Omega_{a,b,c}$ if $e\in \gamma_{\sigma,\tau}$ then,
there must be  $i-t\leq j\leq i-1$ and $i+1\leq k\leq i+d$ such that $\gamma_{\sigma \tau}$ corresponds to $\Ho_{j}^k(A,B)$ or
to $\J_{j}^{i+d+1}(A,p')$. Thus, the congestion placed on $e$ only by paths in $\Gamma_{\Ho}$ is:

$$
  \frac{\sum_{\{\sigma,\tau\vert e \in \gamma_{\sigma\tau}\in \Gamma_{\Ho}\}}
 \vert\gamma_{\sigma,\tau}\vert  {\cal C}(\sigma,\tau)}{{\cal C}(e)}=\sum_{j=i-t}^{i-1}\sum_{k=i+1}^{i+d} \frac{\vert\gamma_{\sigma,\tau}\vert (s_B/s_A)^{i-j} (1+s_B/s_A) }{(1+ (s_B/s_A)^{k+1-j})}
$$
$$ 
\begin{array}{ll}
\hspace{5cm}& \leq 2(d{+}t) \sum_{j'=1}^{t} \sum_{k'=1}^{d}\frac{(s_B/s_A)^{j'} (s_B/s_A) }{(1+(s_B/s_A)^{j'+k'})}\\
&\leq 2t(d{+}t) \sum_{k'=1}^{d} \frac{{s_B}/{s_A}}{({s_B}/{s_A})^{k'}}\leq n^2.
\end{array}
$$

We can similarly show that the congestion placed on $e$ by $\Gamma_{\J}$ is less than $n^2$, where $n $ is the length of the
arrangements or total number of particles. For each $e\in E({\cal M})$, let ${\cal A}_e$ be the congestion $\Gamma$ places on $e$.
We will have ${\cal A}_e\leq 2n^2$.  We also know $\pi_{min}\geq q^{n(n+1)}$, where $q=\max\{s_A/s_B,s_B/s_C\}$ (we consider $q$
to be a constant). Hence, employing the Comparison Theorem we have

$$
t(\EX_3(n)) \leq   \max_{e\in E({\cal M})}{\cal A}_e \pi_{min} t(\EX_t(n))
\leq (2n^4 ) t(\EX_t(n)).
$$
\end{proof}

Having the above connection it suffices to bound the mixing time of $\EX_t(n)$. In the rest of this section our goal is to bound
$t(\EX_t)$. We use the path congestion theorem, which is stated below.  In particular,
for any two arbitrary states $\sigma,\tau \in \Omega_{a,b,c}$,  we introduce a path $\gamma_{\sigma,\tau}$.
Then we show that none of the edges of the Markov chain $\EX_t(n)$ is congested heavily by these paths.
Formally, we employ Theorem \ref{canonicalPaths} and in Theorem \ref{big} we show that $t\left(\EX_t(n)\right) \leq{\cal O}(n^4)$.

\newpage
\begin{thm}(Canonical Paths Theorem \cite{Permanent})\label{canonicalPaths}

Let $\cal M$ be a Markov chain with stationary distribution $\pi$ and $E$  the set of the edges in its underlying graph. 
For any two states $\sigma$ and $\tau$ in the state space $\Omega$ we define a path $\gamma_{\sigma,\tau}$.
The congestion factor for any edge $e\in E$ is denoted by $\Phi_e$ and is defined by
 $
\Phi_e=\frac{1}{C(e)} \sum_{\substack{x,y\\ e\in \gamma_{x,y}} }\pi(x) \pi(y)
$. We can bound the mixing time of $\cal M$ using the congestion factor:
$
t_{\epsilon}({\cal M})\leq 8 \Phi^2 (\log \pi_{min}^{-1} + \log \epsilon ),
$
where $\Phi=\max_{e\in E} \phi_{e}$, $\pi_{min}=\min_{x\in \Omega}\pi(x)$ and $\epsilon$ is the convergence parameter. 
\end{thm}

\textbf{The Paths.}
For each $\sigma,\tau\in \Omega_{a,b,c}$, we introduce the following path in $\EX_t$ from $\sigma$ to $\tau$: We partition $\sigma$
and $\tau$ to $b+1$ blocks; the end points of these blocks are locations of $B$s in $\tau$. For instance if in $\tau$, the first $B$
is located at position $i$ and the second $B$ is located at position $j$ then, the first block in both $\sigma$ and $\tau$ is $[1,i]$,
and the second is $[i+1,j]$. Starting from the first block, we change each block in two steps, first we use Jump moves and change
the relative position of $A$ and $C$s in $\sigma$ to become in the order in which they appear in $\tau$. Then, we bring the $B$
in that block to its location in $\tau$. Formally, we repeat the following loop:
\medskip

\textbf{Notation.} By saying $k=B_j(\sigma)$, we mean the $j$th copy of particle $B$ is located at position $k$ in $\sigma$. 

\noindent 
Starting from $\sigma$, we repeat the following steps until $\tau$ is reached.

Initially, let $i,j=1$. 
\begin{enumerate}
\item Let $k=B_j(\tau)$. We define the  $j$th \emph{block} of $\sigma$ and $\tau$ to be the substring starting from $i$ and ending in $k$.
Note that in $\tau$, each blocks starts right after a $B$ and ends with a $B$. In the $j$th iteration, the goal is to change
$\sigma[i,k]$ until $\sigma[1,k]=\tau[1,k]$, i.e. the first $j$ blocks equal in $\sigma$ and $\tau$.

\item Using Jumps, and starting from the lowest index $i$, we bring particles $C$ or $A$ down until $A$ and $C$ particles in the block
$[i,k]$ have the same order in $\sigma$ and $\tau$.
\item We use Hops and bring the $j$th $B$ in $\sigma$ to $B_j(\tau)$. In this process, we may need to bring several copies of particle
$B$ out of the $j$th block in $\sigma$. In that case, we choose a random ordering of $B$s and move them with respect to that order
(details explained in the proof of Claim \ref{PathClaim}).
\item Set $i=B_j(\tau)+1$.
\item Increment $j$.
\end{enumerate}

\begin{figure}[h]
\noindent
\begin{tiny}

\begin{tabular}{ccccccccr}
\mbox{1st iteration:} &\quad&\quad&\quad\quad\quad\quad\quad\quad\quad\quad\quad\quad\quad\quad\quad\quad\quad\quad\quad&\mbox{2nd iteration:}&\quad&3rd iteration:\\
\end{tabular}

\hspace*{-0.3cm}
\begin{tabular}{rcccccccccccccc}
C&&C&& C&&C&&C&&C&&\textcolor{red}{C}&&A\\
$\sigma:\quad$A& &A&& A&&A&&A&&{A}&&\textcolor{red}{A}&$\quad\quad\tau:$&C\\
B&&B& &B&&B&&B&&\textcolor{red}{B}&&{C}&&C\\
C&&C&&\textcolor{red}{ C}&&C&&C&&\textcolor{red}{C}&&B&&B\\
A&\mbox{\tiny Jump}&A&\mbox{\tiny Jump}&\textcolor{red}{ A}&\mbox{\tiny Jump}&A&\mbox{\tiny Hop}&\textcolor{red}{A}&\mbox{\tiny Hop}&B&\mbox{\tiny Hop}&B&\mbox{\tiny Jump}&B\\
{C}&$\longrightarrow$&\textcolor{red}{C}&$\longrightarrow$& A&$\longrightarrow$&\textcolor{red}{A}&$\longrightarrow$&\textcolor{red}{B}&$\longrightarrow$&A&$\longrightarrow$&A&$\longrightarrow$&A\\
B&\mbox{\tiny Step 1}&B&\mbox{\tiny Step 1}& B&\mbox{\tiny Step 1}&\textcolor{red}{B}&\mbox{\tiny Step 2}&A&\mbox{\tiny Step 2}&A&\mbox{ \tiny Step 2}&A&\mbox{\tiny Step 1}&A\\
\textcolor{red}{C}&& \textcolor{red}{A}&&C&&C&&C&&C&&C&&C\\
\textcolor{red}{A}&& C&&C &&C&&C&&C&&C&&C 

\end{tabular}
\caption{We use the path congestion technique to bound $t(\EX_t)$. In each iteration we fix a block in $\sigma$ until $\tau$ is reached.}
\end{tiny}
\end{figure}


\begin{claim}\label{PathClaim}
 Let $\{\gamma_{\sigma,\tau}\}_{\sigma,\tau\in\Omega_{a,b,c}}$ be the set of paths defined as above. Then, for any arbitrary edge $e$ in
the Markov chain $\EX_t$ the congestion $\Phi_e$, defined in Theorem \ref{canonicalPaths} satisfies $\Phi_e\leq n$.

\end{claim}

We present a roadmap to the proof of the above claim before providing details.
\smallskip

In order to verify the claim, we analyze the congestion of  Jump and Hop edges separately. In both of the analyses, we consider an
edge $e=(\alpha,\beta)$, and to any $\sigma, \tau$ such that $e\in\gamma_{\sigma,\tau}$ we assign a ${\cal F}_e(\sigma,\tau)\in \Omega_{a,b,c}$.
The reverse image of $\cal F$ could be a subset of $\Omega_{a,b,c}\times \Omega_{a,b,c}$. However, using
$q-$binomials\footnote{More information about $q$-binomials can be found in Richard Stanley's course ``Topics in Algebraic Combinatorics,''
Chapter 6 (see \cite{Stanley}). }
we show that $\sum_{\sigma,\tau \text{ are mapped to the same } \zeta} \pi(\sigma)\pi(\tau)$ is bounded by a polynomial function of
$n$ multiplied by $\pi(\zeta)$, and then we conclude the claim. A key factor of our analysis is the use of $q$-binomials.
Note the following observations:
Assume that we have no copies of particle $A$, $b$ copies of $B$, and $c$ copies of particle $C$. Let $M\in \Omega_{0,b,c}$ be the
arrangement with maximum stationary probability, i.e. $M=B^bC^c$. Note that for each $\sigma\in \Omega_{0,b,c}$, $\pi(\sigma)/\pi(M)=(s_{B}/s_C)^t$,
where $t$ is the number of transpositions needed to get from $M$ to $\sigma$. For a constant $t$, the number of $\sigma$s requiring $t$
transpositions is equal to the number of integer partitions of $t$ fitting  in an $b\times c$ rectangle (see Figure \ref{qbin}). Thus:
$$\sum_{\sigma\in \Omega_{0,b,c}} \frac{\pi(\sigma)}{\pi(M)}={{b+c}\choose{b}}_{q}~; q=s_{B}/s_{C}.$$ 
\begin{figure}[!h]

$$\hspace{1.6cm}9=4+4+1\hspace{2.4cm}9=4+3+2\hspace{2.6cm}9=3+3+3$$
\[ {
\tiny
\begin{tabular}{c c}
$\tau_1:$\hspace{1cm}&\\
&{\tiny C}\\
&{\tiny B}\\
&{\tiny C}\\
&{\tiny C}\\
&{\tiny C}\\
&{\tiny B}\\
&{\tiny B}
\end{tabular}\hspace{0.5cm}
\begin{ytableau}
 *(yellow) &  & & \\
 *(yellow) &  *(yellow)  & *(yellow)    & *(yellow)  \\
 *(yellow) & *(yellow)   &  *(yellow) & *(yellow)  \\
\end{ytableau}
\begin{tabular}{c c}
$\hspace{1cm}\tau_2:$\hspace{1cm}&\\
&{\tiny C}\\
&{\tiny C}\\
&{\tiny B}\\
&{\tiny C}\\
&{\tiny B}\\
&{\tiny C}\\
&{\tiny B}
\end{tabular}\hspace{0.5cm}
\begin{ytableau}
 *(yellow) &  *(yellow)& & \\
 *(yellow) &  *(yellow)  & *(yellow)    &   \\
 *(yellow) & *(yellow)   &  *(yellow) & *(yellow)  \\
\end{ytableau}
\begin{tabular}{cc}
$\hspace{1cm}\tau_3$ :\hspace{1cm}&\\
&{\tiny C}\\
&{\tiny C}\\
&{\tiny C}\\
&{\tiny B}\\
&{\tiny B}\\
&{\tiny B}\\
&{\tiny C}
\end{tabular}\hspace{0.5cm}
 \begin{ytableau}
 *(yellow) &  *(yellow)& *(yellow)& \\
 *(yellow) &  *(yellow)  & *(yellow)    &   \\
 *(yellow) & *(yellow)   &  *(yellow) &   \\
\end{ytableau}
}
\]

\caption{Correspondence of partition functions with q-binomials: There are three integer partitions of 9 that fit into a 3$\times$4 rectangle,
and there are three arrangements of gladiators in $\Omega_{0,3,4}$ with $q(\tau_1)=q(\tau_2)=q(\tau_3)=q^9$. In other wors, the coefficient
of $q^9$ in ${{7}\choose{3}}_q$ equals 3.}\label{qbin}
\end{figure}
We will use the following lemma in our proof:
\begin{lemma}\label{qlemma}
 If $q<1/2$ then, $
{{m}\choose{r} }_q<\prod_{i=1}^{r} {1} /{(1-q)}< 2^r < (\frac{1}{q})^r$.
\end{lemma}
\begin{proof}
$$
{{m}\choose{r} }_q=\frac{(1-q^m)(1-q^{m-1})\dots (1-q^{m-r+1})}{(1-q)(1-q^2)\dots (1-q^r)}=\prod_{i=1}^{r} {(1-q^{m-i+1})} /{(1-q^i)}.
$$
$$
{{m}\choose{r} }_q<\prod_{i=1}^{r} {1} /{(1-q)}< 2^r < (\frac{1}{q})^r.
$$
\end{proof}

\begin{proof}[Proof of Claim~\ref{PathClaim}]

Consider an edge $e=(\alpha,\beta)$ corresponding to ${\J_{k}^{k+g}(C,A)}$. Assume that $k=C_{l}(\alpha)$, $k+d=A_m(\alpha)$
(remember the notation $k=p_{m}(\sigma)$ meaning the $m$th copy of particle $p$ is located at position $k$ in $\sigma$),
i.e., this edge is swapping the $l$th $C$ with the $m$th $A$ in $\alpha$.

It follows from the  way we set the paths that, for some $j$, $A_{m}(\alpha)\leq  j< A_{m{+}1}(\alpha), ~ A_{m}(\sigma)=j$,
and  for some $i , A_{m{-}1}(\beta) < i \leq A_{m}{(\beta)}, ~ A_{m}(\tau)=i$. 
The preceding blocks of $\alpha$ have been changed in accordance with $\tau$, and the succeeding blocks of $\alpha$ have
not been changed yet, hence they resemble $\sigma$ blocks.
Therefore we have $\alpha[1,i{-}1]=\tau[1,i{-}1]$ and $\alpha[j{+}1,n]=\sigma[j{+}1,n]$ (see Figure \ref{Move2fig}).

We define the function ${\cal F}_e:\Omega_{a,b,c}\times \Omega_{a,b,c}\rightarrow\Omega_{a,b,c}$ as follows: 
For any $\sigma,\tau$ satisfying $ e\in \gamma_{\sigma,\tau}$, let $\xi_{\sigma,\tau}:=\sigma[1,i-1]\vert\tau[i,n]$
(the symbol $\vert$ denotes concatenation). Since the arrangements of particles is changing, we may have
$\xi_{\sigma,\tau}\notin\Omega_{a,b,c}$. For instance we may have $\tau[i,n]\in \Omega_{x,y,z}$ and
$\sigma[1,i-1]\in \Omega_{x',y',z'}$ but $x+x\neq a$ or $y+y'\neq b$ or $z+z'\neq c$.
However, we know $a-(x+x')+(b-(y+y'))+(c-(z+z'))=0$, which means there is a way to substitute the particles
in $\sigma[1,i-1]$ to change $\xi$ to $\zeta$ so that $\zeta\in\Omega_{a,b,c}$.
We call this stage the substitution stage, in which we identify the particle or particles with extra copies in $\sigma[1,i-1]$,
and we substitute the lowest copies of them with inadequate particles and produce $\zeta\in \Omega_{a,b,c}$.
Then, we define ${\cal F}_{e}(\sigma,\tau):=\zeta$.
For instance, if  $a-(x+x')+(b-(y+y'))=-(c-(z+z'))$, then substitute the lowest $c-(z+z')$ copies of $A$ and
$B$ with $C$s, and produce ${\cal F}_{e}(\sigma,\tau)=\zeta$.
The substitution stage will cause a \emph{substitution cost}, we denote the substitution cost by $\emph{co}(\zeta)$,
and define it as: $\emph{co}(\zeta)=\pi(\zeta)/\pi(\xi)$, where  $\xi=\sigma[1,i-1]\vert\tau[i,n]$. Note that
if we make $t$ substitutions, the substitution cost is at most $(s_C/s_A)^t$.
To make the analysis simpler we only analyze the worst case in which we assume we have substituted $t$
$C$s with $A$s in $\sigma[1,i-1]$. This assumption also means that in $\sigma[i,j]$ we have $t$ more $A$s and $t$ fewer $C$s than in $\alpha[i,j]$.  

\begin{figure}[!ht]

\centerline{\resizebox{0.7\linewidth}{!}{\includegraphics{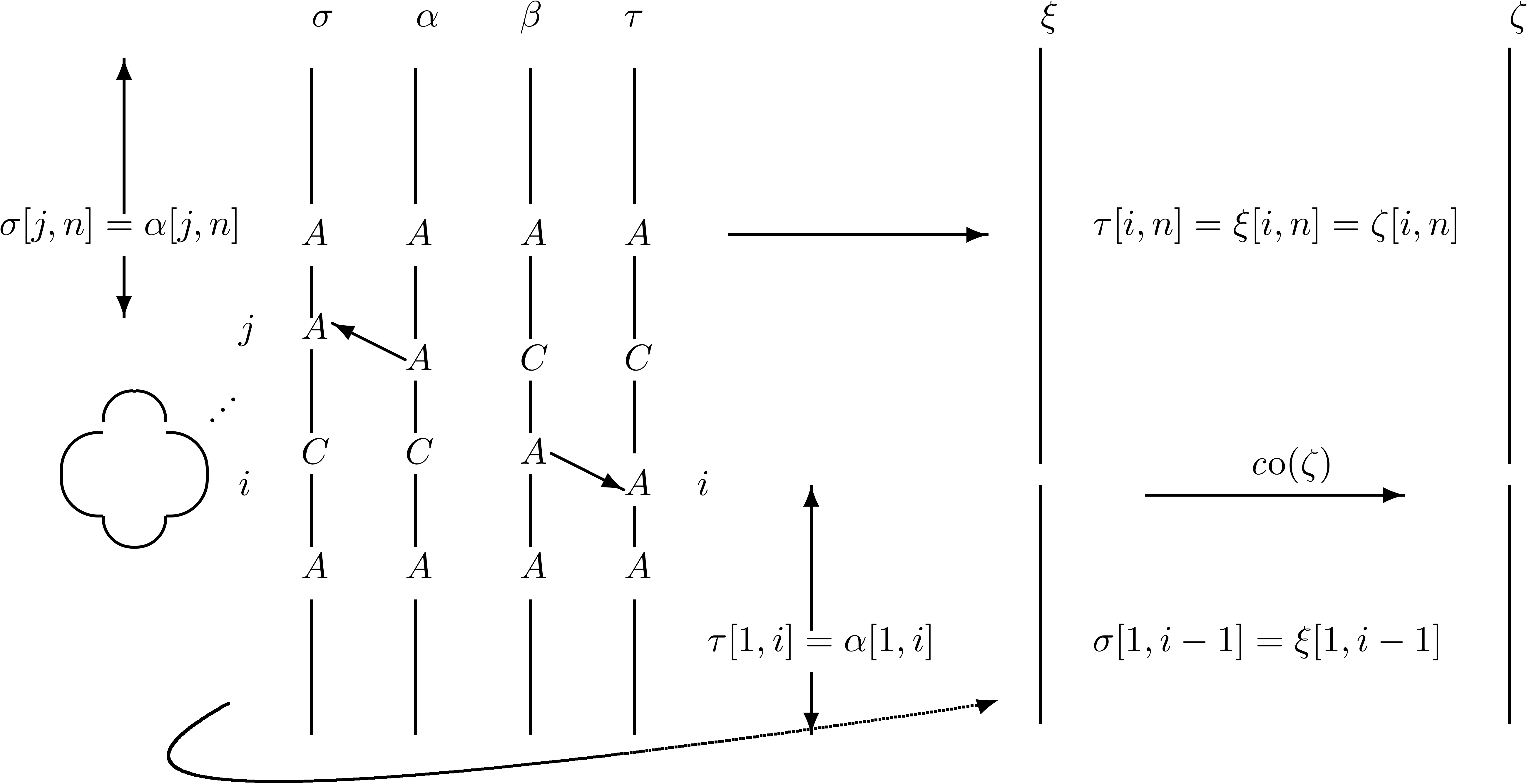}}}

\caption{We define ${\cal F}_{e}(\sigma, \tau)=\zeta$. To produce $\zeta$ we first concatenate $\sigma[1,i-1]$ and $\tau[i,n]$, then
substitute some particles.}\label{Move2fig}
\end{figure}

Consider $\sigma, \tau$ such that $e\in \gamma_{\sigma,\tau}$. Let ${\cal F}_e(\sigma,\tau)=\zeta$. We have, 
$$\frac{\pi(\zeta)}{\pi(\alpha)}= \left(\frac{\pi(\tau)}{\pi(\alpha)}\right)\left(\frac{\pi(\sigma)}{\pi(\alpha)}\right)
\left(\frac{w^i(\alpha[i,j])}{w^i(\sigma[i,j])}\right)
\emph{co}(\zeta),
$$
where the later term is the substitution cost, and $w^i(\sigma[i,j]):= \prod_{\substack{k=i}}^{j} s(k)^{i+\sigma^{-1}(k)}$.
Having $g=A_m(\alpha)-C_l(\alpha)$ we will get:

$$
\Phi_{e}= \left(1+(s_{A}/s_{C})^g\right)\left( \displaystyle\sum_{\substack{\sigma;\\ \alpha[j{+}1,n]=\sigma[j{+}1,n]}} \frac{\pi(\sigma)}{\pi(\alpha)} \sum_{\substack{\tau;\\ \alpha[1,i{-}1]=\tau[1,i{-}1]}} \frac{\pi(\tau)}{\pi(\alpha)} \right)\pi(\alpha)
$$

Let ${\cal S}_t$  be the set of all $\sigma$s with $t$ substitutions. We  have:
\begin{equation}\label{eq3}
\Phi_{e}\leq \sum_{\substack{\zeta \text{ needs }t \\\text{ substititions}}} \frac{1}{\emph{co}(\zeta)} \displaystyle \sum_{\tau}\sum_{\sigma\in S_t}\left( \frac{\pi({\cal F}_e(\sigma,\tau))}{\pi(\alpha)}\right)
\left( \frac{w^i(\sigma[i,j])}{w^i(\alpha[i,j])}\right){\pi(\alpha)}.
\end{equation}

Let $M_t(\alpha)$  be the  arrangement that we get from replacing the lowest $t$ copies of particle $C$ with copies of particle $A$
in  $\alpha[i,j]$. We have:
$
\sum_{\sigma\in S_t} \frac{w^i(\sigma[i,j])}{w^i(\alpha[i,j])}= \frac{w^i(M_t)Q^i_{\bar{B}}(M_t(\alpha))}{w^i(\alpha[i,j])}
$, where $w^i(\sigma[i,j]):= \prod_{\substack{k=i}}^{j} s(k)^{i+\sigma^{-1}(k)}$, and 
$Q^i_{\bar{B}}(M_t(\alpha))
:=\sum_{\substack{\sigma: \text{fix the positions of all Bs}\\ \text{in $M_t(\alpha)$ and rearrange}\\ \text{ the rest of particles }
}} \pi(\sigma)  /\pi(M_t(\alpha)).
$

\begin{figure}[!ht]

\centerline{\resizebox{0.4\linewidth}{!}{\includegraphics{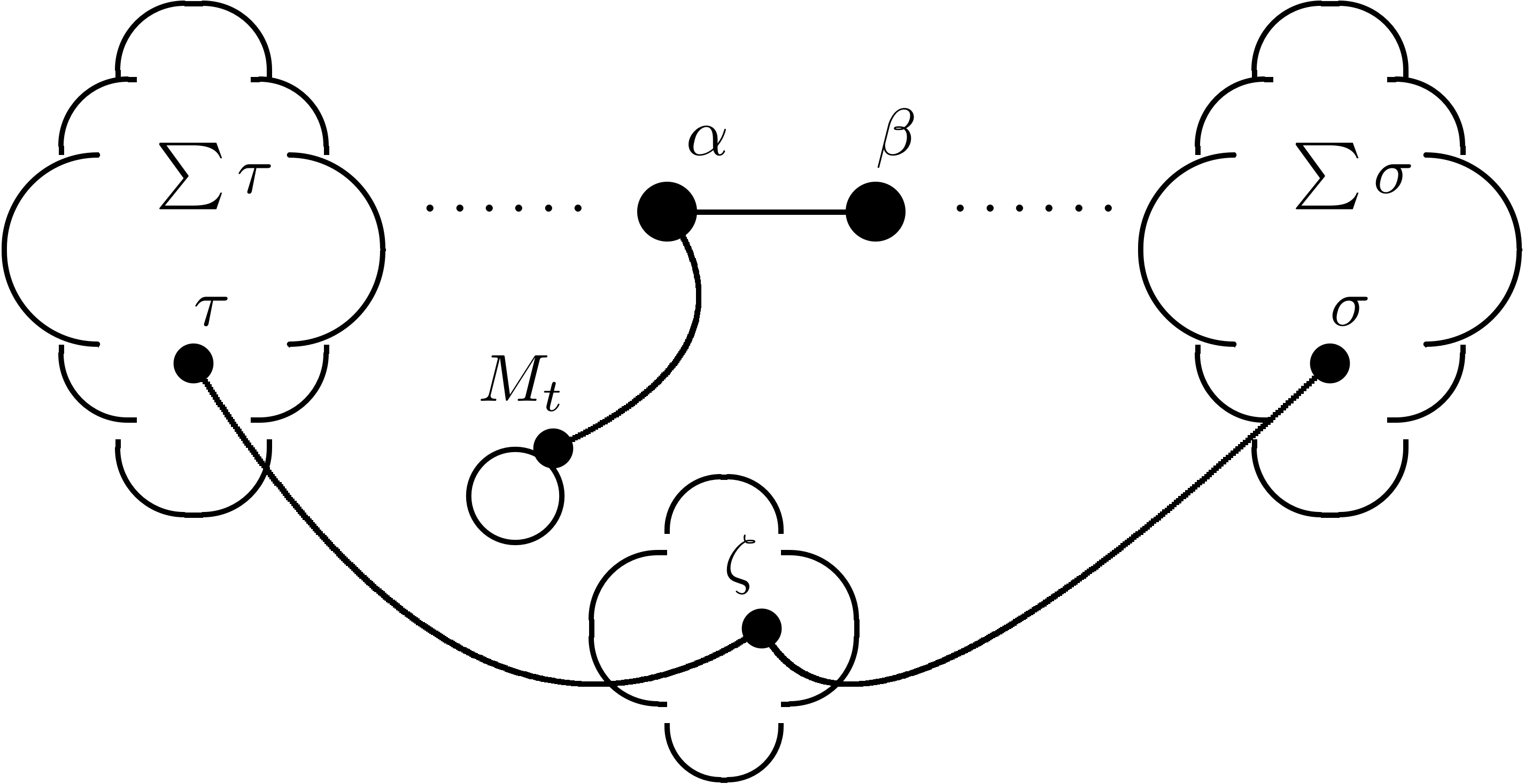}}}

\caption{We obtain $M_t(\alpha)$ from $\alpha$ and then take the sum over all $\zeta$s (Equation \ref{eq3}). }
\end{figure}\label{fig6}
\smallskip
Note that $w^i(M_t)Q_{\bar{B}}(M_t)\leq q^{t(t+1)-2t} w^i(\alpha[i,j])$,  $q$ being $\max\{s_A/s_B,s_B/s_C\}$.
This inequality holds because $Q_{\bar{B}}(M_t)\leq {{y}\choose{t}}_{s_A/s_C}\leq q^{-2t} $ and $w(M_t)/ w(\alpha[i,j])\leq
q^{t(t+1)}$(See Figure 6).

Moreover, $ \sum_{\substack{\zeta \text{ needs }t \\\text{ substititions}}}\frac{1} {\emph{co}(\zeta)}\leq
{{t+b'}\choose{t}}_{q^2}\leq q^{-2t} $, where $b'$ is the number of $B$s in $\sigma[0,i-1]$ and $q=\max\{s_A/s_B, s_B/s_C\}$.

\smallskip
Putting all of the above inequalities together, we will have that each edge of Jump is only congested by:
$$\Phi_e\leq (1+q^g)\sum_t (q^{t(t+1)-4t})\leq n.$$

So far, we showed that any Jump edge is only congested by a factor of a polynomial function of $n$. 
Consider  an edge corresponding to a Hop, namely $e$. We denote this edge by $e=(\alpha,\beta)$. Assume we are swapping $A$ and $B$.

\medskip

Consider a state $\sigma$ traversing $e$ to get to $\tau$, and assume we traversed $e$ while fixing block $[i,j]$.
Since we are making a Hop,  $A$s and $C$s in the block are fixed according to $\tau$, and we are  bringing the $k$th  $B$ to its position in $\tau$.

\medskip

Before we proceed to the proof there is a subtlety about using a Hop that needs to be explained. If $A_k$ has to go down to reach
its position in $\tau$  or if there is only one copy of it in the block there is no complication. 
Let's assume we have $t$ copies of particle $B$ in $\sigma[i,j]$. All of the $t$ copies of $B$ should move up and stand out
of block $\sigma[i,j]$ to reach their position in $\tau$. In order to accomplish this, we choose a subset $S$ of
$\{1_k,\dots 1_{t{+}k}\}$ uniformly at random and we move the elements of $S$ in decreasing order of their index out of the block.  

\smallskip

Assume, when going from $\sigma$ to $\tau$ we used  $e=(\alpha,\beta)$ and in $\alpha[i,j]$ we have $t$ copies of particle $B$:
$B_k, \dots , B_{k{+}t}$ and swapping  $B_{k+l}, B_{k+l+1},\dots , B_{k+d}$  with the next $A$.  We have,
$\tau[1,i]=\alpha[1,i]$, $\sigma[j+t,n]=\alpha[j+t,n]$, and for any $i$, if
$ B_{k+i}(\alpha)<B_{l+k}(\alpha) \text{ then, }  B_{k+i}(\alpha)=B_{k+i}(\sigma)$. 
The following information about $S$ can be determined by examining $\alpha$ and $\beta$: $B_{k+d+1},\dots B_{k+t}\notin S$
while $S$ may contain any of $B_{k},\dots B_{k+l}$. Therefore, among the random paths connecting $\sigma$ to $\tau$,
there are $2^l$ subsets traversing through $e$ and hence the congestion they place on 
$e$ is $\pi(\tau)\pi(\sigma)/2^{t-l}$.

 \smallskip

To bound  $\Phi_{e}$ for each $e$ we introduce correspondence 
${\cal F}_e:\Omega_{a,b,c}\times \Omega_{a,b,c}\rightarrow \Omega_{a,b,c}$  satisfying:
\begin{equation}\label{eq4} \hspace{2.5cm}
\forall \zeta\in {\cal F}_e(\Omega_{a,b,c}); ~\frac{\sum_{\substack{\sigma,\tau;\\ {\cal F}^{-1}_{e}(\zeta)=(\sigma,\tau)}} \pi (\sigma)\pi(\tau) }{\pi(\alpha)}\leq 2^{t-l}  \pi(\zeta); 
\end{equation}
where $c$ is the number of $C$s in $\alpha[i,j]$ and ${\cal F}_e (\sigma,\tau) \neq \mbox{NULL} \text{ if and only if, }
e=(\alpha,\beta) \in \gamma_{\sigma,\tau}.$

Let $\sigma$ and $\tau$ be two ends of a path traversing through  $e$. We define ${\cal F}_e := \sigma[1,i-1]\vert \tau[i,n]$;
to verify Equation \ref{eq4}, take $\zeta= {\cal F}_e(\sigma,\tau)$. We have
$
\frac{\pi(\sigma)\pi(\tau)}{\pi(\alpha)}= \frac{\pi(\sigma)}{\pi(\alpha)}\frac{\pi(\tau)}{\pi(\alpha)}\pi(\alpha).
$
Thus,
$$
\begin{array}{ll}
 \frac{\pi(\zeta)}{\pi(\alpha)} =&\frac{\pi(\zeta[1,i-1])}{\pi(\alpha[1,i-1])} \frac{\pi(\zeta[i,j-1])}{\pi(\alpha[i,j-1])} \frac{\pi(\zeta[j,n])}{\pi(\alpha[j,n])} =\frac{\pi(\sigma[1,i-1])}{\pi(\alpha[1,i-1])} \frac{\pi(\tau[i,j-1])}{\pi(\alpha[i,j-1])} \frac{\pi(\tau[j,n])}{\pi(\alpha[j,n])}\\
\hspace{0.75cm}=&\frac{\pi(\sigma')}{\pi(\sigma)} ~ \frac{\pi(\sigma)}{\pi(\alpha)} \frac{\pi(\tau)}{\pi(\alpha)},
\end{array}$$
where  $\sigma'$ is the following arrangement: $\sigma':=\alpha[1,i-1]\vert\sigma[i,j-1]\vert\alpha[j,n]$.
We have $\pi(\sigma')/\pi(\alpha)=\pi(\sigma[i,j])/\pi(\alpha[i,j])$. Hence, 
$$
\sum_{\sigma,\tau; {\cal F}(\sigma,\tau)=\zeta} \frac{\pi(\sigma)\pi(\tau)}{\pi(\alpha)}= 
\sum_{\substack{\sigma,\tau\\ {\cal F}(\sigma,\tau)=\zeta}} \frac{\pi(\sigma')}{\pi(\sigma)} \pi(\zeta).
$$

Since we have $t-l$ $B$s with undecided position between $j{-}i$ other elements we have $\sum \frac{\pi(\sigma')}{\pi(\sigma)} \leq {{j-i+t-l}\choose{t-l}}_q$, where $q=\max\{s_A/s_B, s_B/s_C\}$. Thus, we have $\sum \frac{\pi(\sigma')}{\pi(\sigma)}  \leq 2^{t-l} $.  Hence, the congestion placed on $e$ is:

$$
 \Phi_{e=(\alpha,\beta)}=  (1+q^g)\sum_{\substack {\sigma,\tau\\ e\in \gamma_{\sigma,\tau}}} \frac{\pi(\sigma)\pi(\tau)}{\pi(\alpha) 2^{t-l}} \leq 1.
$$

Summing up, we showed that for any edge $e$, $\Phi_{e}\leq \max \{n,1\}$.

 \end{proof}
 
Having the above claim, we now use the path congestion Theorem (Theorem \ref{canonicalPaths}) to bound $t(\EX_t(n))$:
 \begin{thm}\label{big}
If $s_A/s_B,s_B/s_C\leq 1/2$, then $t(\EX_t(n))\leq {\cal O}(n^4)$. 
\end{thm}
\begin{proof}

Since $\pi_{min}\geq q^{n(n+1)}$, $q$ being maximum of $s_A/s_B$ and $s_B/s_C$, we can  apply Theorem \ref{canonicalPaths} and we will have,
$
t_{\epsilon}({\EX}_{t})\leq 8n^2 (n^2+ \ln (\epsilon^{-1}))
\implies
t({\EX}_{t})\leq 8n^4.
$

\end{proof}

Finally, from Lemma \ref{lem1} and Theorem \ref{big} we conclude Theorem \ref{mainthm}.

\section{League Hierarchies for Trenary Trees.}\label{trees}
 
As mentioned earlier in Section \ref{intro}, the \emph{league hierarchies} are a class of monotone adjacent transposition Markov chains and
were introduced by Bhakta et al.\ (SODA 2014) \cite{Dana} for binary trees. Here we extend their definition to trenary trees. 
   
\begin{defin}\textbf{(League hierarchies for ternary trees).} \label{LeagueHiDef}
Consider a ternary tree whose leaves are labeled by $1,2,\dots n$, and inside each interior node $v$ there are three numbers
$S_{R,v}>S_{C,v}> S_{L,v}$ satisfying :
$\frac{S_{R,v}}{S_{C,v}}, \frac{S_{C,v}}{S_{L,v}}>1/2$.
We define the Markov chain $\Lc$ on $S_n$ as follows.  In state $\sigma\in S_n$, $1\leq i \leq n$ is chosen uniformly at random and
$\sigma_i$ and $\sigma_{i+1}$ are swapped with probability $p_{\sigma(i), \sigma(i+1)}$. The $p_{i,j}$s are defined in accordance with
the tree structure.  For each  $i\leq j$, the probability if swapping $p_{i,j}$ is equal to $g$; $S_{X,v}/(S_{X,v}+S_{Y,v})$, $v$
being the lowest common ancestor of the leaves labeled  $i$ and $j$, and $X,Y$ being one of $R,L,C$ depending respectively on whether
they are in the right, left or central subtree rooted at $v$.
 \end{defin}

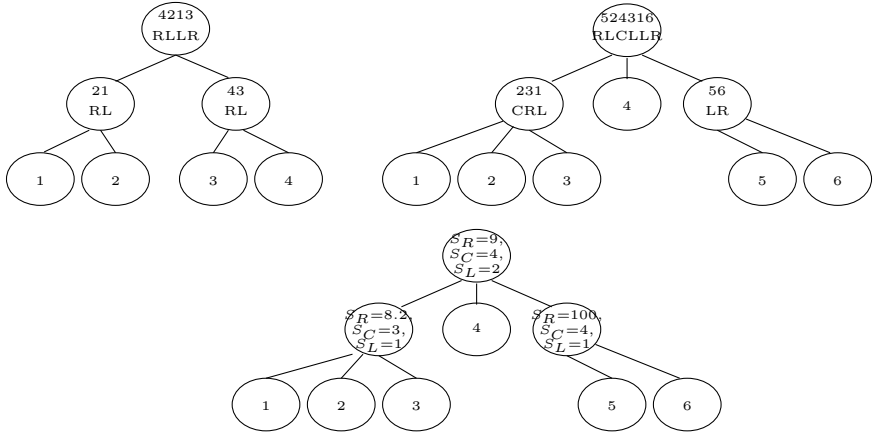
\begin{figure}[!h]\label{figlatticepath}
\begin{tikzpicture}\label{fig7}

\draw (4,4) ellipse (4.5mm and 3.5 mm);
\node[above] at (4,4) {\tiny 4213};
\node[below] at (4,4.1) {\tiny RLLR};

\draw (3,3) ellipse (4.5mm and 3.5 mm);
\node[above] at (3,3) {\tiny 21};
\node[below] at (3,3.1) {\tiny RL};

\draw (4.8,3) ellipse (4.5mm and 3.5 mm);
\node[above] at (4.8,3) {\tiny 43};
\node[below] at (4.8,3.1) {\tiny RL};

\draw (4.5,2) ellipse (4.5mm and 3.5 mm);
\node[above] at (4.5,1.8) {\tiny 3};

\draw (5.5,2) ellipse (4.5mm and 3.5 mm);
\node[above] at (5.5,1.8) {\tiny 4};

\draw (2.2,2) ellipse (4.5mm and 3.5 mm);
\node[above] at (2.2,1.8) {\tiny 1};

\draw (3.2,2) ellipse (4.5mm and 3.5 mm);
\node[above] at (3.2,1.8) {\tiny 2};

\draw (4.7,3.35) -- (4,3.65);
\draw (3.2,3.3) -- (4,3.65);

\draw (5.5,2.35) -- (4.9,2.65);
\draw (4.5,2.35) -- (4.7,2.65);
\draw (3.2,2.35) -- (3,2.65);
\draw (2.25,2.36) -- (2.85,2.65);

\draw (7.2,2) ellipse (4.5mm and 3.5 mm);
\node[above] at (7.2,1.8) {\tiny 1};

\draw (8.2,2) ellipse (4.5mm and 3.5 mm);
\node[above] at (8.2,1.8) {\tiny 2};

\draw (9.2,2) ellipse (4.5mm and 3.5 mm);
\node[above] at (9.2,1.8) {\tiny 3};

\draw (10,3.98) ellipse (4.5mm and 3.5 mm);
\node[above] at (10,3.97) {\tiny 524316};
\node[below] at (10,4.1) {\tiny RLCLLR};

\draw (10,3) ellipse (4.5mm and 3.5 mm);
\draw (10,3.33) -- (10,3.61);
\node[above] at (10,2.8) {\tiny 4};

\draw (11.8,2) ellipse (4.5mm and 3.5 mm);
\node[above] at (11.8,1.8) {\tiny 5};

\draw (12.8,2) ellipse (4.5mm and 3.5 mm);
\node[above] at (12.8,1.8) {\tiny 6};

\draw (8.7,3) ellipse (4.5mm and 3.5 mm);
\node[above] at (8.7,3) {\tiny 231};
\node[below] at (8.7,3.1) {\tiny CRL};

\draw (11.2,3) ellipse (4.5mm and 3.5 mm);
\node[above] at (11.2,3) {\tiny 56};
\node[below] at (11.2,3.1) {\tiny LR};

\draw (9,3.3) -- (9.8,3.65);
\draw (11,3.3) -- (10.2,3.65);

\draw (7.2,2.35) -- (8.35,2.77);
\draw (8.2,2.35) -- (8.49,2.7);
\draw (9.2,2.35) -- (8.7,2.65);
\draw (11.8,2.35) -- (11.2,2.65);
\draw (12.7,2.35) -- (11.58,2.8);


\draw (5.2,-1) ellipse (4.5mm and 3.5 mm);
\node[above] at (5.2,-1.2) {\tiny 1};

\draw (6.2,-1) ellipse (4.5mm and 3.5 mm);
\node[above] at (6.2,-1.2) {\tiny 2};

\draw (7.2,-1) ellipse (4.5mm and 3.5 mm);
\node[above] at (7.2,-1.2) {\tiny 3};

\draw (8,0.98) ellipse (4.5mm and 3.5 mm);
\node[below] at (8,0.2) {\tiny 4};

\draw (8,0) ellipse (4.5mm and 3.5 mm);
\draw (8,0.33) -- (8,0.61);

\draw (9.8,-1) ellipse (4.5mm and 3.5 mm);
\node[above] at (9.8,-1.2) {\tiny 5};

\draw (10.8,-1) ellipse (4.5mm and 3.5 mm);
\node[above] at (10.8,-1.2) {\tiny 6};

\draw (6.7,0) ellipse (4.5mm and 3.5 mm);

\draw (9.2,0) ellipse (4.5mm and 3.5 mm);

\node[above] at (9.2,-0.04) {${\scriptscriptstyle S_{R}=100,}$}; 
\node[below] at (9.2,0.18) {${\scriptscriptstyle S_{C}=4,}$};
\node[below] at (9.2,0.0) {${\scriptscriptstyle S_{L}=1}$};
\node[above] at (8,0.97) {${\scriptscriptstyle S_{R}=9,}$}; 
\node[below] at (8,1.2) {${\scriptscriptstyle S_{C}=4,}$};
\node[below] at (8,1.0) {${\scriptscriptstyle S_{L}=2}$};
\node[above] at (6.7,-0.04) {${\scriptscriptstyle S_{R}=8.2,}$}; 
\node[below] at (6.7,0.18) {${\scriptscriptstyle S_{C}=3,}$};
\node[below] at (6.7,0.0) {${\scriptscriptstyle S_{L}=1}$};

\draw (7,0.3) -- (7.8,0.65);
\draw (9,0.3) -- (8.2,0.65);

\draw (5.2,-0.65) -- (6.35,-0.33);
\draw (6.2,-0.65) -- (6.49,-0.33);
\draw (7.2,-0.65) -- (6.7,-0.35);
\draw (9.8,-0.65) -- (9.2,-0.35);
\draw (10.7,-0.65) -- (9.58,-0.2);

\end{tikzpicture}
\caption{The  league hierarchies for binary trees and the ternary trees.
}
\end{figure}

When the tree is a binary tree, Bhakta et al.\ (SODA 2014) \cite{Dana}, analyzed the mixing time of the league hierarchies using the comparison
theorem and by comparing it to the following Markov chain:

\smallskip

Consider a tree whose leaves are labeled by $1,2,\dots n$, and inside each interior node $v$ there are three numbers $S_{R,v}>S_{C,v}> S_{L,v}$
satisfying $\frac{S_{R,v}}{S_{C,v}}, \frac{S_{C,v}}{S_{L,v}}>1/2$. The Markov chain $\M_{tree}$ works as follows:
in state $\sigma\in S_n$, $i$ and $j$ with $1\leq i<j \leq n$ are chosen uniformly at random. Let $v=i\wedge j$ be the lowest common ancestor
of $i$ and $j$, and  $T_v$ the subtree rooted at $v$. We swap $\sigma(i)$ and $\sigma(j)$ iff $\sigma(i+1),\sigma(i+2),\dots, \sigma(j-1)\notin T_v$,
with the probabilities given in Definition \ref{LeagueHiDef}
 
Bhakta et al.\ proved that if the tree is a binary tree then $\M_{tree}$ is rapidly mixing. In the following lemma we extend their result:  
 
\begin{lemma} For the ternary tree labeled as in Definition \ref{LeagueHiDef}, $t (\M_{tree})\leq  {\cal O}(n^{10}\log n)$.
 
\end{lemma}

\begin{proof}
The  Markov chain of our discourse is a product of $n-1$ smaller three particle systems (See Figure 7). Thus, by
Theorem \ref{product} and \ref{mainthm} we conclude the result. 
\end{proof}

Bhakta et al.\ \cite{Dana} had shown that the two Markov chains $\Lc$ and $\M_{tree}$ have the same stationary distribution.
To conclude  Corollary \ref{LeagueHi}, it remains to show that their mixing time is related. As in Section \ref{proof}, we use the comparison technique.

\begin{lemma}
$t(\Lc)\leq n^4 t(\M_{tree})$.
\end{lemma}

\begin{proof}

We label the interior nodes of the tree as done in Figure 7, then each edge in $\M_{tree}$ will be corresponding to the exchange
of two particles of the same type in the root between which there is no particle of the same type.
Consider an arbitrary edge $e=(\sigma,\tau)\in \M_{tree}$, we correspond the path $\Gamma_{e}$ lying on $\Lc$
whose construction will be explained in the next paragraph. As an example,  assume we are swapping $5$ and $6$ in $587231964$
in a full balanced binary tree with $9$ leaves and thus labeled by $CRRLLLRCC$. We  define the path between any arbitrary
$\sigma$ and $\tau$ in $\Lc$ in which the particles at positions $i$ and $j$ are swapped as follows: 
Do the following until $\sigma(i)$ and $\sigma(j)$ meet (we call this stage one):
\begin{enumerate}
\item If $\sigma(i+1)=\sigma(j-1)$, swap $\sigma(i)$ and $\sigma(i+1)$, then swap $\sigma(j-1)$ and $\sigma(j)$. Then, repeat. 

In the example starting from $\sigma=587231964; C{\color{red}R}RLLL{\color{red}R}CC$, the first edges in the path will be
$587231964\rightarrow {\color{red}85}7231964\rightarrow 857231{\color{red}69}4$. 
\item If $\sigma(i+1)\neq \sigma(j-1)$, then swap $\sigma(i)$ and $\sigma(i+1)$ if  $\sigma(i) > \sigma(i+1)$.
Otherwise, swap $\sigma(j)$ and $\sigma(j-1)$ if  $\sigma(j) > \sigma(j-1)$. In case none of the above holds,
we can conclude that both $\sigma(i)$ and $\sigma(j)$ are labeled by $C$ and between them we have a sequence of $L$ and $R$s.
In this case, using adjacent transpositions, take   $\sigma(k)$ to the position of $i+1$ if $\sigma(i+1)$ is labeled with $R$,
and $k$ is the smallest index greater than $i+1$ so that $\sigma(k)$ is labeled by $L$. Then, repeat. Otherwise, find a $k$
with the largest index smaller than $j-1$ that is labeled by $R$ and take it to the position at $j-1$. Then, repeat. 

In our example we continue by the following edges: 

$857231694 ;R C{\color{red}R}LL{\color{red}L}CRC \rightarrow 852{\color{red}7}31694 \rightarrow 8523{\color{red}7}1694 
\rightarrow 85231{\color{red}7}694.  $
Then, we restart: $852317694; R CLLLCRC \rightarrow 852317694    \rightarrow 825316794  \rightarrow 823561794  $.
\end{enumerate}

Repeat the above steps until $\sigma(i)$ and $\sigma(j)$ meet, then swap them, and enter stage two.

In our example, we swap $5$ and $6$:  $ 823{\color{red}65}1794\rightarrow 823{\color{red}65}1794;$ $ RLLCCLRRC$.

\medskip

Let $\mu$ be the permutation in which $\sigma(i)$ and $\sigma(j)$ meet. Note that by the transpositions of stage one, on our way from
$\sigma$ we only visit arrangements $\mu$, we only visit permutations $\omega$ satisfying $\pi(\sigma)>\pi(\omega)$.
After reaching this state, in stage two, we take $\sigma(i)$  to  $j$ and $\sigma(j)$ to $i$; and also, in the case where the two particles
were $C$s, potentially take back other particles to their original positions. Note that by  $\sigma(i)$ can reach $j$ by exactly performing the transpositions 
$\sigma(j)$ made in stage one, and vice versa. 
Thus, for any $\omega$  visited on this stage we always will have: $ \pi(\sigma)\left(\max\{\frac{s_R}{s_C},\frac{s_C}{s_L}\}\right)\geq\pi(\omega)$.

\smallskip

In our example, in the second stage we take $6$ and $5$ to the final positions and also $7$ to its original position:

\smallskip

 $823{\color{red}65}1794 \rightarrow 82{\color{red}6}31{\color{red}5}794;  RLCLLCRRC\rightarrow  82{\color{red}6}317{\color{red}5}94\rightarrow  82{\color{red}6}3179{\color{red}5}4\rightarrow 8{\color{red}6}23179{\color{red}5}4 \rightarrow 8{\color{red}6}23719{\color{red}5}4\rightarrow 8{\color{red}6}2719{\color{red}5}4.  $

\medskip

The congestion placed on each edge in $\M_{tree}$ by the paths of $\gamma$ will be bounded by:

$$
 {\cal A}_e= \frac{\sum_{\{\sigma,\tau\vert e \in \gamma_{\sigma\tau}\in \Gamma_{\Ho}\}}
 \vert\gamma_{\sigma,\tau}\vert  {\cal C}(\sigma,\tau)}{{\cal C}(e)}\leq \sum_{i=1}^{n}\sum_{j=1}^nq  \leq q n^3.\hspace{3cm}
 $$
 $$
\hspace{7cm} \text{\tiny ;where } q=\max\{\frac{s_R}{s_C},\frac{s_C}{s_L}\} \text{\tiny is a constant}.
$$

Using the comparison theorem, we complete the proof.

$$
t(\Lc) \leq   \max_{e\in E({\cal M})}{\cal A}_e \pi_{min} t({\cal{M}}_{tree})\leq n^4 t({\cal{M}}_{tree}).
$$
\end{proof}
\subparagraph*{Acknowledgements.}

We would like to thank Dana Randall for a very helpful conversation about the gladiator problem and Fill's conjecture,
and Sergi Elizalde for his help and knowledge concerning generating functions.

\end{document}